\documentclass[11pt]{article}
\usepackage[margin=1in]{geometry}
\usepackage{amsmath,amssymb,amsthm}
\usepackage{graphicx}
\usepackage{caption}
\usepackage{booktabs}
\usepackage{microtype}
\usepackage{cite}
\usepackage{xcolor}
\definecolor{linkcol}{RGB}{20,50,120}
\usepackage[colorlinks=true,linkcolor=linkcol,citecolor=linkcol,urlcolor=linkcol]{hyperref}

\graphicspath{{figures/}}

\theoremstyle{plain}
\newtheorem{theorem}{Theorem}

\newtheorem{proposition}{Proposition}
\theoremstyle{definition}

\newtheorem{example}{Example}
\newtheorem{remark}{Remark}

\newcommand{\dd}{\,d}
\newcommand{\Qbar}{\bar{Q}_n}

\newcommand{\Qinf}{\bar{Q}_{\infty}}
\newcommand{\gap}{\mathcal{G}}

\DeclareMathOperator*{\argmax}{arg\,max}

\begin{document}

\title{A Quantization Problem Posed by Adaptive Streaming}

\author{Yuriy A. Reznik\\[2pt]
\normalsize Massachusetts Institute of Technology, Cambridge, MA, USA\\
\normalsize \texttt{yreznik@mit.edu}}
\date{August 28, 2026}

\maketitle

\begin{abstract}
In adaptive bitrate (ABR) streaming, the delivery technology behind most Internet video, each title is encoded at several bitrates, forming an \emph{encoding ladder} of \emph{renditions}. Each client plays the highest rendition that its network bandwidth can sustain. We show that choosing the ladder is a problem of \emph{scalar quantization} of the bandwidth distribution. However, this quantization problem is of an unusual kind: the client's logic pins each quantization cell's reproduction value to the cell's \emph{left edge}, and the distortion measure is a one-sided quality loss rather than a squared error. The constraint reshapes the classical theory. The \emph{quality gap} of a ladder is the distance between its average delivered quality and the \emph{quality limit} that infinitely many renditions would attain. For optimal $n$-rendition ladders the gap $\gap_n^*$ decays as $\Theta(1/n)$, not the textbook $\Theta(1/n^2)$. Moreover, with large $n$, $n\,\gap_n^*\to C_w=\frac12(\int\!\sqrt{p\,Q'}\,\dd B)^2$, with the limit attained when the ladder's rates follow the density $\sqrt{p\,Q'}$, which replaces Panter--Dite's $p^{1/3}$; here $p$ is the bandwidth density and $Q$ the quality--rate curve of the content. The law yields closed-form design rules: quantile rung placement, the rendition count $n(\varepsilon)\approx C_w/\varepsilon$ needed to reach tolerance $\varepsilon$, and the economic ladder size $n^*=\sqrt{\lambda C_w/\kappa}$ when a rendition costs $\kappa$ to operate and the quality gap is priced at $\lambda$. The same analysis extends to the design of ladders employing different video resolutions, codecs, perceptual quality metrics, and the two-dimensional adaptation logic of modern web players.
\end{abstract}

\section{Introduction}
\label{sec:intro}

Most videos on the Internet---movies, sports, news, user-generated clips---are delivered using \emph{adaptive bitrate} (ABR) streaming protocols, such as HLS and DASH~\cite{Pantos2017,DASH2012}. The idea is simple. The capacity of the network path between a service and a viewer is unknown in advance and may be varying, so the service prepares \emph{several} encodings of the same content at different bitrates and lets each receiving device choose, moment by moment, the one its connection can sustain. The set of prepared encodings is the \emph{encoding ladder}; its entries are \emph{renditions}, informally \emph{rungs}.

The ladder is thus a finite set of operating points standing in for a continuum of network conditions, and someone must decide what this set should be. How many renditions? At which bitrates? In current practice, the question ``what is the best ladder?'' has never lacked answers---standing recommendations such as Apple's HLS authoring guidelines~\cite{AppleTN2224}, per-title and context-aware optimizers built into commercial encoders~\cite{Netflix2015,Chen2018,Reznik2020SMPTE}, and a growing body of work casting the search as a learning problem~\cite{Katsenou2021,Telili2025}. What none of this gives is an understanding of the problem itself: what mathematical structure it has, what an optimal ladder looks like, and how far from optimal any given design sits.

The objective of this paper is to develop such understanding. What we find is that the mathematics behind the ladder design problem is, recognizably, \emph{quantization}~\cite{GrayNeuhoff1998}. A ladder assigns to every possible network bandwidth $B$ a single delivered encoding; that is, it \emph{quantizes} the continuous random variable $B$ into $n$ representative operating points. Once this identification is made precise, the tools of quantization theory become available: Lloyd--Max-type optimality conditions~\cite{Lloyd1982,Max1960}, dynamic-programming design with global optimality guarantees, and, most usefully, high-resolution asymptotics that describe optimal ladders in closed form.

The identification comes with a twist, and the twist is what makes the problem interesting. A streaming client is \emph{conservative}: it plays the highest-bitrate rendition whose bitrate fits \emph{below} its available bandwidth. In quantization language, the reproduction value of each cell is not free to sit at the cell's centroid, as in the textbook problem; it is pinned to the cell's \emph{left edge}. This small change has large consequences. Quantization errors within a cell now all have one sign and never cancel, the within-cell loss is first-order in the cell width rather than second-order, and the whole high-resolution picture changes: the optimal quality gap decays as $\Theta(1/n)$ instead of the classical $\Theta(1/n^2)$, and the optimal density of the ladder's rates becomes $\sqrt{p\,Q'}$, replacing the Panter--Dite density $p^{1/3}$~\cite{Bennett1948,PanterDite1951}.

\subsection{Related prior work}
The models and the design objective come from our earlier papers~\cite{Reznik2018PV,Reznik2021PCS,Reznik2021DCC}. Ladder design as maximization of average delivered quality, solved by numerical search, was formulated in~\cite{Reznik2018PV}; the fitted models used here, together with exact computed designs, were developed in~\cite{Reznik2021DCC}. Those papers, however, stopped at formulation and numerical optimization. They did not take the next step: a rigorous analysis of the problem as a mathematical one. The present paper takes that step, and the tools for it already exist in two bodies of prior work.

The first body of work is classical quantization theory. The high-resolution analysis of nonuniform quantizers begins with Bennett in 1948~\cite{Bennett1948} and Panter and Dite in 1951~\cite{PanterDite1951}. Its rigorous form is Zador's limit theorem from 1963~\cite{Zador1982,BucklewWise1982}, extended to general source distributions by Graf and Luschgy in 2000~\cite{GrafLuschgy2000}. The optimality conditions and the descent algorithm they suggest are due to Lloyd in 1957~\cite{Lloyd1982} and Max in 1960~\cite{Max1960}. Dynamic programming for globally optimal scalar quantizers goes back to Bruce in 1965~\cite{Bruce1965}, with a matrix-search speedup by Wu in 1991~\cite{Wu1991}. A comprehensive account of the field is Gray and Neuhoff's survey~\cite{GrayNeuhoff1998}. This literature assumes free reproduction points, which settle in the interior of their cells with error mass on both sides, and so does not cover the edge-constrained class, where all error mass lies on one side.

The second is a set of earlier one-sided problem formulations. The \emph{assortment} (``standard sizes'') \emph{problem} of operations research is exactly edge-constrained quantization under one-sided linear loss. From it comes exact design by dynamic programming, used by Sadowski already in 1959~\cite{Sadowski1959}, with later finite-$n$ methods surveyed in~\cite{Pentico2008}. \emph{Approximation theory} holds the deterministic core of the problem: best one-sided approximation of an increasing function by an $n$-step staircase with free knots. Its ``balanced error'' asymptotics---the error decays as $\Theta(1/n)$, with knots placed to equalize the error of every step---were established by Burchard and McClure in the mid-1970s~\cite{Burchard1974,McClure1975} and are surveyed in~\cite{DeVore1998}. \emph{Discrete-rate link adaptation} performs the same edge-constrained selection over a fading channel, with average throughput as the objective and thresholds optimized numerically for small fixed rate sets, from the late 1990s on~\cite{GoldsmithChua1997,ChungGoldsmith2001}. None of these efforts, however, developed the asymptotics of the problem in the probabilistic form that ladder design needs.

\subsection{Contributions}
The problem class is thus old, appearing in several fields. What this paper adds is a connection between this class and adaptive streaming. We make four contributions.
\emph{(i) The identification} (Section~\ref{sec:quant}): ladder design is scalar quantization of the bandwidth distribution, with each cell's reproduction value pinned to the cell's left edge. The design therefore has $n$ degrees of freedom---the cell boundaries---not the $2n$ of the classical quantizer, where reproduction points are chosen separately. Over any $m$-point discretization of the rate axis---and rates are discrete in practice---the globally optimal ladder is computable by dynamic programming in $O(nm^2)$ operations. This places ABR ladder design inside the old class---and, in return, gives the class a fresh, measurement-rich application.
\emph{(ii) The high-resolution solution of the problem} (Proposition~\ref{thm:law}, Section~\ref{sec:hr}): writing $\gap_n$ for the quality gap of an $n$-rendition ladder, $n\cdot\inf\gap_n \to C_w = \frac12(\int\sqrt{p\,Q'}\,\dd B)^2$, achieved by rates with empirical density $\propto\sqrt{p\,Q'}$. This is a translation of a classical approximation-theory result to the present setting: we state it in Zador's probabilistic form, with the explicit constant, and give a short self-contained proof. Its hypotheses ask only for a density and a nondecreasing quality curve, so it applies verbatim to the relatives above.
\emph{(iii) Closed-form design rules} (Section~\ref{sec:rules}): quantile rung placement that lands within a few percent of optimal at practical sizes, the count $n(\varepsilon)\approx C_w/\varepsilon$ of renditions needed to come within $\varepsilon$ of the quality limit, and the economically optimal count $n^*=\sqrt{\lambda C_w/\kappa}$ when a rendition costs $\kappa$ to operate and the quality gap is priced at $\lambda$. All are verified against exact optima on models fitted to mass-scale streaming data.
\emph{(iv) Extensions} (Section~\ref{sec:ext}): the analysis transfers, with structure intact and constants recomputed, to ladders that also choose encoding resolutions, to perceptual quality models, and to the two-dimensional adaptation of modern web players, where an exact reduction collapses the enlarged design problem into a chain of one-dimensional ones.

\section{System and Models}
\label{sec:model}

\begin{figure}[t]
\centering
\includegraphics[width=0.97\textwidth]{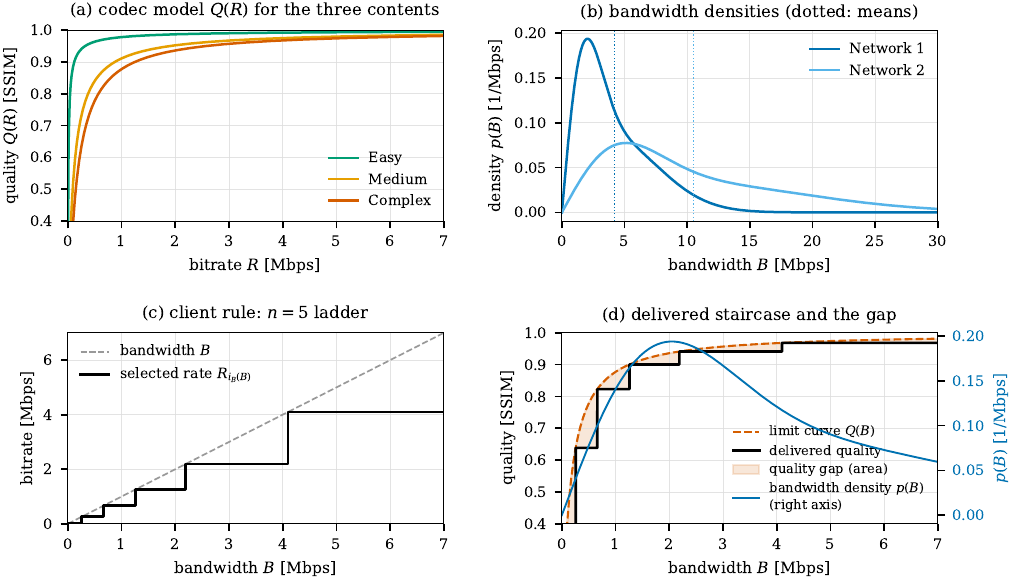}
\caption{The three system models, in the fitted forms of Table~\ref{tab:models}. (a) Quality--rate model for the three test sequences at 1080p. (b) The two bandwidth densities. (c) The client rule~\eqref{eq:client} for a five-point ladder (rates $257$, $659$, $1257$, $2187$, $4097$ kbps); below $R_1$ the client buffers. (d) The delivered-quality staircase pressed against the curve $Q(B)$; the shaded area, weighted by the bandwidth density (right axis), is the quality gap~\eqref{eq:gap}.}
\label{fig:models}
\end{figure}

\subsection{The mechanics of ABR streaming}

A video asset is transcoded into $n$ independent encodings (renditions), each characterized by a bitrate $R_i$ (kbps) and, in general, an encoding resolution. Until Section~\ref{sec:ext} we will assume that all renditions are encoded at a single resolution (e.g., 1080 lines), so a ladder is a list of bitrates $0<R_1<\dots<R_n$. Each rendition is cut into short, aligned, independently decodable \emph{segments}, so a client may fetch consecutive segments from different renditions. The client estimates its end-to-end throughput $B$ by timing its own segment downloads and requests, every few seconds, the next segment from the rendition its current estimate allows. This throughput fluctuates, so the client protects playback: it keeps a buffer of downloaded segments and selects the next ones conservatively.

The operator controls the ladder; it does not control the clients, the network, or the audience, but it observes all three through playback analytics. The design problem thus has a classic shape: choose a finite codebook (the ladder) to maximize an objective (average delivered quality) over a source distribution (bandwidth), under a fixed usage rule (the client logic).

\subsection{The three models}
\label{sec:models}

Analysis requires three ingredients: a model of the encodings achievable on the asset, a model of client selection, and a model of the network bandwidth distribution. We define them generically, as functions with only the key properties required for the analysis. For numerical examples, we use more specific models proposed in~\cite{Reznik2021DCC,Reznik2021PCS}; Fig.~\ref{fig:models} and Table~\ref{tab:models} summarize their forms, parameters, and origin.

\begin{table}[!t]
\centering
\caption{The models used in Fig.~\ref{fig:models} and all numerical examples.}
\label{tab:models}
\footnotesize
\begin{tabular}{@{}p{0.155\textwidth}p{0.45\textwidth}p{0.33\textwidth}@{}}
\toprule
Model & Form used in the examples & Origin and fitting \\
\midrule
Quality--rate $Q(R)$ & $Q(R)=\bigl(1+(R/\rho_c)^{-\gamma_c}\bigr)^{-1/\gamma_c}$, with quality measured against the source by the structural-similarity index SSIM~\cite{WangBovik2004}. Three 1080p test sequences: ``Easy'' (animated cartoon; $\rho_c=4.2$, $\gamma_c=0.746$), ``Medium'' (natural scene; $\rho_c=84.6$, $\gamma_c=0.959$), ``Complex'' (crowd run; $\rho_c=154.2$, $\gamma_c=1.034$); $R$ in kbps. & Fitted through a few dozen probe encodes per sequence; this model was shown to predict SSIM of H.264-encoded video reasonably well~\cite{Reznik2021DCC}. \\[2pt]
Bandwidth density $p(B)$ & Mixture of two Rayleigh densities, $p(B)=\omega f(B\,|\,\sigma_1)+(1-\omega)f(B\,|\,\sigma_2)$ with $f(B\,|\,\sigma)=(B/\sigma^2)\,e^{-B^2/(2\sigma^2)}$ and $\omega=0.4287$. ``Network 1'': $\sigma_1=1802$, $\sigma_2=4499$; ``Network 2'': Network 1 scaled $2.5\times$, i.e., $\sigma_1=4506$, $\sigma_2=11248$; $B$ in kbps. & Fitted in~\cite{Reznik2021DCC} to observed CDFs of TCP traffic bandwidth in wireless networks under different loads. \\
\bottomrule
\end{tabular}
\end{table}

\subsubsection{Codec and content}
We assume that the encoding quality $Q(R)$, produced by an encoder operating on a given video sequence at a target bitrate $R$, can be modeled as a function
\begin{equation}
\begin{gathered}
Q(R)\in[0,1],\quad R\in\mathbb{R}_+,\quad \text{$C^1$ and strictly increasing in $R$},\\
Q(0)=0,\qquad Q(\infty)=1.
\end{gathered}
\label{eq:codec}
\end{equation}
Specific examples of such functions, obtained by fitting experimentally measured data, are shown in Fig.~\ref{fig:models}(a) and explained in Table~\ref{tab:models}.

\subsubsection{The client}
We model the client's time-averaged behavior by a deterministic threshold rule: at bandwidth $B$, a client selects rendition
\begin{equation}
i_B(B) = j \;\Longleftrightarrow\; R_j \le B < R_{j+1},
\label{eq:client}
\end{equation}
for $j=0,1,\dots,n$, with the conventions $R_0=0$ and $R_{n+1}=\infty$: the highest-rate rendition whose bitrate does not exceed $B$. The index $j=0$ denotes \emph{buffering}: below $R_1$ no rendition is sustainable, and the delivered quality is $Q(R_0)=0$. An example showing the application of this rule with a five-point ladder is given in Fig.~\ref{fig:models}(c).

\subsubsection{The network}
We model the bandwidth $B$ by a single continuous distribution on $\mathbb{R}_+$, with a density
\begin{equation}
p(B)\ge0,\qquad \int_0^\infty p(B)\dd B=1,
\label{eq:net}
\end{equation}
and a distribution function $F$. In practice, operators can measure such a distribution by aggregating players' reported download rates. Two examples of such distributions, used in the numerical experiments, are shown in Fig.~\ref{fig:models}(b), with their origin explained in Table~\ref{tab:models}.

\subsection{Average quality, the quality limit, and the gap}

Under~\eqref{eq:client}, rendition $i$ is delivered when $B\in[R_i,R_{i+1})$, so its \emph{load} is $p_i = F(R_{i+1})-F(R_i)$; the remaining mass $p_0=F(R_1)$ is the probability of buffering. The vector $(p_i)$ is the \emph{rendition load distribution}. The central performance measure is the \emph{average delivered quality}, i.e., the average quality delivered by a system using an $n$-point ladder $R_1,\dots,R_n$:
\begin{equation}
\Qbar(R_1,\dots,R_n) = \int_0^\infty\!\! p(B)\, Q\bigl(R_{i_B(B)}\bigr)\dd B
= \sum_{i=1}^n p_i\, Q(R_i).
\label{eq:Qbar}
\end{equation}

As the ladder grows dense, delivered quality approaches $Q(B)$ pointwise and $\Qbar$ approaches the \emph{quality limit}
\begin{equation}
\Qinf \;=\; \int_0^\infty p(B)\,Q(B)\dd B.
\label{eq:Qinf}
\end{equation}
This limit depends on content and network, but not on the ladder. What design can influence is the \emph{quality gap}
\begin{equation}
\gap_n \;=\; \Qinf - \Qbar \;=\; \int_0^\infty\! p(B)\bigl[ Q(B) - Q(R_{i_B(B)})\bigr]\dd B,
\label{eq:gap}
\end{equation}
the $p$-weighted area between the curve $Q(B)$ and the staircase of delivered quality $Q(R_{i_B(B)})$ inscribed under it.

As an example of all these terms at work, consider the models and the five-point adaptation ladder presented in Fig.~\ref{fig:models}. The computation of the quality gap is illustrated in Fig.~\ref{fig:models}(d). For this ladder, the average delivered quality is $\bar{Q}_5=0.9223$ SSIM, the quality limit is $\Qinf=0.9460$, and the quality gap is $\gap_5=0.024$.

\section{Ladder Design is Edge-Constrained Quantization}
\label{sec:quant}

The ladder design problem is to find a set of rates $\{R_1^*,\dots,R_n^*\}$ such that the average quality delivered by the streaming system is maximal:
\begin{equation}
\{R_1^*,\dots,R_n^*\} \;=\; \argmax_{\substack{R_1<\dots<R_n\\[1pt] R_i\in[R_{\min},R_{\max}]}} \;\Qbar(R_1,\dots,R_n).
\label{eq:design}
\end{equation}
The objective $\Qbar$ is the average delivered quality~\eqref{eq:Qbar}. The first condition under the $\argmax$ orders the rates; the second confines them to a range $[R_{\min},R_{\max}]$ of practically usable rates. By~\eqref{eq:gap}, maximizing $\Qbar$ is the same as minimizing the quality gap $\gap_n$---in quantization terms, minimizing the average distortion. As posed, the problem is nonconvex and nonsmooth (each $R_i$ appears both inside $Q$ and inside integration limits), which is why it was previously attacked by generic numerical search~\cite{Reznik2018PV,Reznik2021PCS,Reznik2021DCC}. Its actual structure is far better than generic.

\subsection{The identification, and two amendments}

Table~\ref{tab:dict} sets out the mapping onto a scalar quantizer: bandwidth is the source, the client rule the partition, delivered quality the reproduction, the quality gap the average distortion, and the ladder the codebook. This is a quantization problem, but not the textbook one: it differs in two respects, and both differences drive all that follows.

\begin{table}[t]
\centering
\caption{The quantization--streaming dictionary.}
\label{tab:dict}
\small
\begin{tabular}{ll}
\toprule
Quantization & ABR streaming \\
\midrule
source variable & network bandwidth $B \sim p$ \\
codebook of size $n$ & encoding ladder $(R_i)_{i=1}^n$ \\
encoder / partition & client rule $i_B(\cdot)$; cell $i$ is $[R_i, R_{i+1})$ \\
reproduction value & delivered quality $Q(R_i)$ \\
nearest-neighbor rule & conservative rule: reproduce at left edge \\
distortion measure & quality loss $Q(B)-Q(R_{i_B(B)})$, one-sided \\
average distortion & quality gap $\gap_n = \Qinf - \Qbar$ \\
Lloyd conditions & fused condition~\eqref{eq:lloyd} \\
\bottomrule
\end{tabular}
\end{table}

\emph{Amendment 1: edge-constrained reproduction.} Classically, cell boundaries $[t_{i-1}, t_i)$ and reproduction points are separate objects, coupled at the optimum by the centroid and midpoint conditions of Lloyd and Max~\cite{Lloyd1982,Max1960}---$2n$ degrees of freedom in all. Here the client's conservatism \emph{fuses} them: cell $i$ is $[R_i, R_{i+1})$, so the reproduction value \emph{is} the cell's left edge:
\begin{equation}
R_i \;=\; t_{i-1}.
\label{eq:edge}
\end{equation}
This holds for every cell, the \emph{buffering cell} $[R_0,R_1)$ included. The design has $n$ degrees of freedom---the rates---not $2n$, and within-cell errors all have one sign, so the loss grows linearly with cell width instead of quadratically. Figure~\ref{fig:binning} contrasts the two situations.

\emph{Amendment 2: quality loss, not squared error.} The distortion measure is not a squared distance between $B$ and its reproduction but the one-sided quality loss $Q(B)-Q(R_i)$ of Table~\ref{tab:dict}: errors are measured after mapping bandwidths through the quality curve. The local cost of such errors is easy to quantify. Serving sessions near a bandwidth $B$ at a rate $\delta$ below it loses quality $Q(B)-Q(B-\delta)\approx Q'(B)\,\delta$ per session, where $Q'$ is the first derivative of $Q$, and such sessions occur with probability density $p(B)$. Their product, $w(B)=p(B)\,Q'(B)$, is therefore the \emph{loss density}: the rate at which average quality is lost, per unit of under-delivery, around $B$. Quantization is costly where $w$ is large and nearly free where it vanishes.

\begin{figure}[t]
\centering
\includegraphics[width=0.94\textwidth]{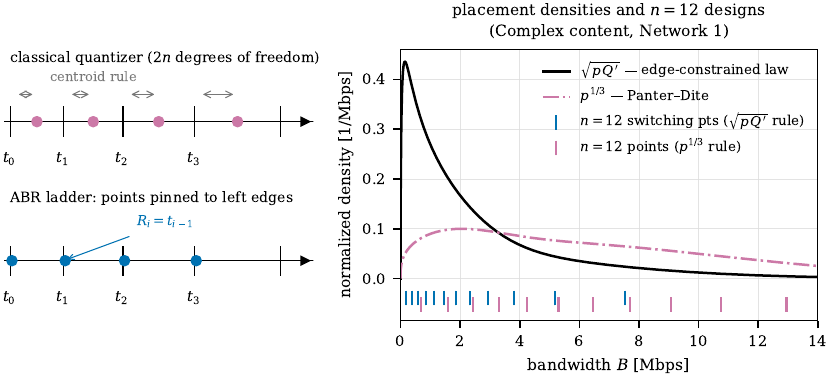}
\caption{Binning in classical quantizers versus ABR ladders. Left: classically, cell boundaries $t_i$ and reproduction points are separate objects; in the ABR ladder each reproduction value is pinned to its cell's \emph{left edge}. Right: on the fitted models (Complex content, Network~1), the classical density $p^{1/3}$ (dash-dotted) follows the bandwidth distribution alone, while the edge-constrained $\sqrt{p\,Q'}$ (solid) discounts saturated bandwidths, skewing the $n=12$ design (tick marks) toward low-to-middle rates.}
\label{fig:binning}
\end{figure}

\subsection{Optimality conditions and exact design}
\label{sec:dp}

Written through the loads, the objective of~\eqref{eq:design} is
\begin{equation}
\Qbar(R_1,\dots,R_n) = \sum_{i=1}^n Q(R_i)\bigl[F(R_{i+1})-F(R_i)\bigr],
\label{eq:Qbars}
\end{equation}
with the conventions of~\eqref{eq:client}, $R_0=0$ and $R_{n+1}=\infty$. Every design variable appears twice: $R_i$ sets the quality of its own cell \emph{and} the boundary shared with the cell below. Setting $\partial\Qbar/\partial R_i=0$:
\begin{equation}
Q'(R_i)\bigl[F(R_{i+1}) - F(R_i)\bigr] \;=\; p(R_i)\,\bigl[ Q(R_i) - Q(R_{i-1}) \bigr].
\label{eq:lloyd}
\end{equation}
This is the Lloyd pair fused into one condition by the edge constraint, and it holds for \emph{every} $i=1,\dots,n$. Its reading: raising $R_i$ gains quality at rate $Q'(R_i)$ for cell $i$, and pushes traffic of mass $p(R_i)\,dR$ down one rung, losing the step $Q(R_i)-Q(R_{i-1})$; at an optimum the two balance. For $i=1$ the step is the whole $Q(R_1)$, so the lowest rate is an interior optimum, needing no external cap.

Condition~\eqref{eq:lloyd}, taken over $i=1,\dots,n$, is a system of $n$ equations in $n$ unknowns, and it can be solved numerically. It is not sufficient, however: since $\Qbar$ is nonconvex, the system generally has multiple solutions---local optima and saddle points among them---and satisfying it does not certify a global optimum. We therefore compute exact solutions of~\eqref{eq:design} differently. Restricting the rates to an $m$-point lattice makes the search space finite, and because the objective~\eqref{eq:Qbars} couples only consecutive pairs $(R_i,R_{i+1})$, the optimum over the lattice can be found by dynamic programming, a technique known in quantizer design since Bruce~\cite{Bruce1965}. Its complexity is $O(nm^2)$ operations, against the $\binom{m}{n}\sim m^n/n!$ placements (for $m\gg n$) that a brute-force search over the same lattice would have to examine. We use this method to compute all optimal ladders in the numerical examples that follow.

\section{The High-Resolution Solution}
\label{sec:hr}

We now turn to the \emph{high-resolution} regime. The assumption is that the number of renditions $n$ is large, so that every cell $[R_i,R_{i+1})$ is narrow and $p$ and $Q'$ are nearly constant across it; the box constraints of~\eqref{eq:design} are dropped, so the rates may lie anywhere in $\mathbb{R}_+$. From this assumption we derive the rate at which the optimal gap decays with $n$, the constant that governs the decay, and the placement density that the optimal rates follow.

\subsection{The mechanism}

Everything rests on an exact identity for the loss of one cell. Writing $Q(B)-Q(R_i)=\int_{R_i}^B Q'(v)\dd v$ and exchanging integrals,
\begin{equation}
L_i = \!\int_{R_i}^{R_{i+1}}\!\!\! p(B)\bigl[Q(B)-Q(R_i)\bigr]\dd B
= \iint_{\mathcal{T}_i} p(B)\, Q'(v)\dd v \dd B,
\label{eq:cellloss}
\end{equation}
the integral of $p\cdot Q'$ over the \emph{triangle} $\mathcal{T}_i=\{R_i\le v\le B\le R_{i+1}\}$ of side $\Delta_i=R_{i+1}-R_i$. If $p$ and $Q'$ are roughly constant across a small cell,
\begin{equation}
L_i = \tfrac12\, w(R_i)\,\Delta_i^2 + o(\Delta_i^2),
\label{eq:cellloss2}
\end{equation}
where $w(R_i)=p(R_i)\,Q'(R_i)$. The classical centroid quantizer's per-cell distortion, by contrast, is $p\,\Delta^3/12$: one power of $\Delta$ is lost because within-cell errors share a sign.

The Bennett-style heuristic is now immediate. Fix $b$ with $1-F(b)$ negligible and restrict the gap~\eqref{eq:gap} to $[0,b]$, at the price of at most $1-F(b)$ error. Describe a large-$n$ ladder by a \emph{point density} $\lambda\ge0$, $\int_0^b\lambda(B)\dd B=1$, so that cells near a point $B$ have width $\approx1/(n\lambda(B))$. Summing~\eqref{eq:cellloss2},
\[
\gap_n \approx \frac{1}{2n}\int_0^b \frac{w(B)}{\lambda(B)}\dd B,
\]
and minimizing $\int w/\lambda$ subject to $\int\lambda=1$ is a Cauchy--Schwarz exercise: the optimum is $\lambda^*\propto\sqrt{w}$, with value $(\int\sqrt{w})^2$. Every claim of this section is already visible: the $1/n$ decay, the $\sqrt{p\,Q'}$ placement rule, and the constant.

\subsection{The $\Theta(1/n)$ law}

The core of the following result is not new: in deterministic form, it is known in approximation theory~\cite{Burchard1974,McClure1975,DeVore1998}. We state and prove it in the probabilistic form that the design problem needs.

\begin{proposition}[$\Theta(1/n)$ law and $\sqrt{w}$ companding; after~\cite{Burchard1974,McClure1975}]
\label{thm:law}
Let $p$ be a probability density and $Q$ nondecreasing, both $C^1$ on $[0,b]$, and let $w=p\,Q'$. Then
\begin{equation}
\lim_{n\to\infty} n\cdot \inf_{R}\,\gap_n(R) \;=\; C_w \;=\; \frac12\Bigl( \int_0^b \sqrt{w(B)}\dd B \Bigr)^{\!2},
\label{eq:law}
\end{equation}
where the infimum is over sequences $0=R_0<R_1<\dots<R_n\le b$ with $n$ free points, and the limit is achieved by sequences whose points have empirical density proportional to $\sqrt{w}$.
\end{proposition}

\begin{proof}
Write $W=\int_0^b\sqrt{w(B)}\dd B$, so $C_w=W^2/2$.

\emph{Upper bound.} Assume first $w>0$ on $[0,b]$. Define the \emph{companding map} $\Lambda(x)=\frac1W\int_0^x\sqrt{w(B)}\dd B$, continuous and strictly increasing onto $[0,1]$, and place
\begin{equation}
R_i=\Lambda^{-1}\bigl( i/(n{+}1) \bigr), \qquad i=1,\dots,n.
\label{eq:quantile}
\end{equation}
Each of the $n+1$ cells then carries equal $\sqrt{w}$-mass: $\int_{R_i}^{R_{i+1}}\sqrt{w}=W/(n{+}1)$ (with $R_{n+1}=b$). Since $\sqrt{w}\ge\mu>0$, all widths obey $\Delta_i\le W/((n{+}1)\mu)$, shrinking uniformly. The expansion~\eqref{eq:cellloss2} holds uniformly over cells (its $o(\Delta^2)$ is controlled by uniform continuity of $p,Q'$ on $[0,b]$), and equal mass gives $\sqrt{w(R_i)}\,\Delta_i = W/(n{+}1) + o(1/n)$ uniformly, whence $L_i = \frac12 (W/(n{+}1))^2 (1+o(1))$ and, summing the $n+1$ cells, $\limsup_n n\inf_R \gap_n \le C_w$. If $w$ vanishes somewhere, apply the construction on $\{w\ge\epsilon\}$ with $n-o(n)$ points and spend $o(n)$ points keeping widths bounded on the complement, where by~\eqref{eq:cellloss2} the loss is $O(\epsilon)$ per unit length; let $n\to\infty$, then $\epsilon\to0$.

\emph{Lower bound.} Fix a sequence $R$ and let $\Delta_i=R_{i+1}-R_i$. From the exact identity~\eqref{eq:cellloss}, $L_i \ge \frac12\,\underline{w}_i\,\Delta_i^2$ with $\underline{w}_i = \min_{[R_i,R_{i+1}]}p \cdot \min_{[R_i,R_{i+1}]}Q'$. Fix $\eta>0$. If some cell is wider than $\eta$, then $\gap_n(R)\ge c(\eta)>0$ for all $n$ (a wide cell contains a fixed sub-triangle of~\eqref{eq:cellloss} meeting $\{w>0\}$), so such sequences cannot approach the infimum. Otherwise, all $\Delta_i\le\eta$; by uniform continuity there is a modulus $\omega(\eta)\to0$ with $\underline{w}_i \ge w(R_i)-\omega(\eta)$, and by Cauchy--Schwarz,
\[
(n{+}1) \sum_i \underline{w}_i \Delta_i^2 \;\ge\; \Bigl( \sum_i \sqrt{\underline{w}_i}\,\Delta_i \Bigr)^{\!2} \;\to\; \Bigl( \int_0^b\!\sqrt{w} \Bigr)^{\!2}
\]
as $\eta\to0$, since $\sum_i\sqrt{w(R_i)}\,\Delta_i$ is a Riemann sum, over the $n+1$ cells, of mesh at most $\eta$. Hence, with $n/(n{+}1)\to1$, $\liminf_n n\inf_R\gap_n \ge W^2/2$.
\end{proof}

\begin{remark}[Relation to Zador's theorem]
The Panter--Dite formula is itself a heuristic of exactly the kind sketched above. Its rigorous form is Zador's theorem~\cite{Zador1982,BucklewWise1982}, established for general absolutely continuous sources by Graf and Luschgy~\cite{GrafLuschgy2000}. Proposition~\ref{thm:law} is the Zador-form statement for the edge-constrained class, under $C^1$ assumptions on a compact interval. Extending it to Graf--Luschgy generality---notably to bandwidth distributions with atoms, as produced by rate-capped service tiers---remains open.
\end{remark}

\begin{example}[Uniform source, linear quality]
\label{ex:toy}
Let $p=1$ and $Q(B)=B$ on $[0,1]$, with $R_0=0$ pinned. Then $\gap_n = \frac12\sum_i \Delta_i^2$ under $\sum_i\Delta_i=1$, minimized by giving the $n+1$ cells equal widths: $\gap_n^*=\frac{1}{2(n+1)}$, against $1/(12n^2)$ for classical quantization of the same source (Fig.~\ref{fig:toy}). Here $w\equiv1$, so $C_w=\frac12$, and $(n{+}1)\,\gap_n^*$ equals $C_w$ exactly at every finite $n$: the law holds with $n+1$---the number of cells---in place of $n$, an off-by-one that washes out in the limit.
\end{example}

\begin{figure}[t]
\centering
\includegraphics[width=0.55\textwidth]{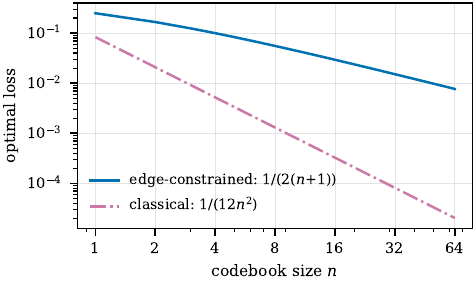}
\caption{Example~\ref{ex:toy}: uniform source, linear quality. The edge constraint alone costs an order of magnitude in the decay exponent.}
\label{fig:toy}
\end{figure}

Table~\ref{tab:compare} compares the two theories side by side, and Fig.~\ref{fig:binning} (right) shows the practical difference: $\sqrt{p\,Q'}$ discounts regions where quality has saturated and pushes rungs toward low-to-middle bandwidths, where $p^{1/3}$, blind to saturation, would keep placing them high.

\begin{table}[t]
\centering
\caption{Classical high-resolution quantization vs.\ the edge-constrained case.}
\label{tab:compare}
\footnotesize
\setlength{\tabcolsep}{3pt}
\begin{tabular}{lll}
\toprule
 & Classical (centroid, sq.\ err.) & ABR ladder (edge-constrained) \\
\midrule
per-cell loss & $p\,\Delta^3/12$ & $w\,\Delta^2/2$, \ $w=p\,Q'$ \\
point density & $\propto p^{1/3}$ (Panter--Dite) & $\propto \sqrt{p\,Q'}$ \\
optimal decay & $\frac{1}{12n^2}\bigl(\int p^{1/3}\bigr)^3$ & $\frac{1}{2n}\bigl(\int\sqrt{p\,Q'}\bigr)^2$ \\
deg.\ of freedom & $2n$ (bound.\ $+$ points) & $n$ (boundaries only) \\
\bottomrule
\end{tabular}
\end{table}

\section{Design Rules and Numerical Verification}
\label{sec:rules}

The law gives the designer three closed-form rules. To verify them, we compare against exact optima: solutions of~\eqref{eq:design} computed by the dynamic program of Section~\ref{sec:dp} on a fine rate lattice, under the models of Section~\ref{sec:model} with SSIM quality---six content--network configurations in all, with Complex content on Network 1 as the running example. The box constraints are dropped throughout, matching the setting in which the law was derived.

\subsubsection*{Rule 1: rung placement without optimization}
The construction inside the proof of Proposition~\ref{thm:law} is the recipe: compute $\Lambda(x)=\frac1W\int_0^x\sqrt{w(B)}\dd B$, where $W=\int_0^b\sqrt{w(B)}\dd B$, and set $R_i=\Lambda^{-1}(i/(n{+}1))$---the \emph{quantile rule}~\eqref{eq:quantile}: the rates split $[0,b]$ into $n{+}1$ cells of equal $\sqrt{w}$-mass. Every ingredient is a one-dimensional integral of known model functions, so the ladder comes out in closed form.

Table~\ref{tab:quantile} tests the rule against the exact optima on both networks, listing the ladders with their delivered qualities and gaps. The rule's rates sit above the optimal ones, and the penalty shrinks as $n$ grows: in quality terms the rule gives up under $0.005$ SSIM from $n=4$ on, and its excess gap keeps falling ($3\%$ of the optimum at $n=16$, $1.3\%$ at $n=32$, on Network 1). The other four configurations behave the same way. From $n\approx8$ on the quantile ladder is, for engineering purposes, the optimal ladder.

\begin{table}[t]
\centering
\caption{The quantile rule~\eqref{eq:quantile} against exact optimal ladders and the law's estimate $C_w/n$ (Complex content, SSIM).}
\label{tab:quantile}
\small
\setlength{\tabcolsep}{4.5pt}
\begin{tabular}{cllccc}
\toprule
$n$ & Ladder & Rates [kbps] & $\Qbar$ & $\gap_n$ & $C_w/n$ \\
\midrule
\multicolumn{6}{l}{\emph{Network 1} ($\Qinf=0.9460$, $C_w=0.1076$)} \\
\midrule
2 & optimal & 535, 1864 & 0.8819 & 0.0640 & 0.0538 \\
  & rule~\eqref{eq:quantile} & 935, 2723 & 0.8629 & 0.0830 & \\
3 & optimal & 386, 1139, 2680 & 0.9047 & 0.0413 & 0.0359 \\
  & rule~\eqref{eq:quantile} & 658, 1646, 3566 & 0.8963 & 0.0496 & \\
4 & optimal & 307, 832, 1699, 3427 & 0.9158 & 0.0302 & 0.0269 \\
  & rule~\eqref{eq:quantile} & 511, 1191, 2226, 4281 & 0.9113 & 0.0346 & \\
5 & optimal & 257, 659, 1257, 2187, 4097 & 0.9223 & 0.0237 & 0.0215 \\
  & rule~\eqref{eq:quantile} & 420, 935, 1646, 2723, 4899 & 0.9196 & 0.0264 & \\
6 & optimal & 223, 548, 1001, 1638, 2614, 4633 & 0.9265 & 0.0194 & 0.0179 \\
  & rule~\eqref{eq:quantile} & 358, 772, 1311, 2043, 3164, 5436 & 0.9247 & 0.0212 & \\
8 & optimal & 178, 414, 718, 1106, 1613, 2299, 3363, 5472 & 0.9317 & 0.0143 & 0.0134 \\
  & rule~\eqref{eq:quantile} & 279, 575, 935, 1381, 1948, 2723, 3936, 6310 & 0.9307 & 0.0152 & \\
\midrule
\multicolumn{6}{l}{\emph{Network 2} ($\Qinf=0.9773$, $C_w=0.052$)} \\
\midrule
2 & optimal & 911, 3852 & 0.9410 & 0.0363 & 0.0260 \\
  & rule~\eqref{eq:quantile} & 1915, 6173 & 0.9213 & 0.0560 & \\
3 & optimal & 634, 2211, 5736 & 0.9549 & 0.0223 & 0.0173 \\
  & rule~\eqref{eq:quantile} & 1289, 3588, 8193 & 0.9468 & 0.0305 & \\
4 & optimal & 492, 1549, 3497, 7479 & 0.9613 & 0.0160 & 0.0130 \\
  & rule~\eqref{eq:quantile} & 968, 2509, 4979, 9906 & 0.9572 & 0.0201 & \\
5 & optimal & 408, 1195, 2510, 4672, 9072 & 0.9650 & 0.0123 & 0.0104 \\
  & rule~\eqref{eq:quantile} & 775, 1915, 3588, 6173, 11405 & 0.9626 & 0.0147 & \\
6 & optimal & 351, 979, 1960, 3435, 5736, 10549 & 0.9673 & 0.0100 & 0.0087 \\
  & rule~\eqref{eq:quantile} & 647, 1542, 2792, 4541, 7232, 12723 & 0.9657 & 0.0115 & \\
8 & optimal & 277, 717, 1353, 2228, 3404, 5034, 7580, 12756 & 0.9700 & 0.0072 & 0.0065 \\
  & rule~\eqref{eq:quantile} & 490, 1106, 1915, 2957, 4313, 6173, 9080, 14906 & 0.9693 & 0.0080 & \\
\bottomrule
\end{tabular}
\end{table}

\subsubsection*{Rule 2: how many renditions}
Inverting $\gap_n\approx C_w/n$: reaching a tolerance $\varepsilon$ of the limit requires $n(\varepsilon)\approx C_w/\varepsilon$. This turns the perennial ``how many rungs?'' question into a calculation.

Table~\ref{tab:cw} shows the constants and the resulting ladder sizes for all six configurations. For Complex content on Network 1 ($C_w=0.108$), coming within $0.01$ SSIM of the limit requires $n\approx11$ renditions; for Easy content ($C_w=0.020$), two suffice. The exact gaps (Fig.~\ref{fig:gapvsn}) confirm both predictions within one rendition.

\subsubsection*{Rule 3: the economic ladder size}
Adding renditions costs money, and so does the quality lost to a coarse ladder. Let $\kappa$ be the operational cost of one rendition (encoding, storage, cache footprint), and let the impact of delivering suboptimal quality be linear in the quality gap, $\lambda\,\gap_n$, with both priced in the same currency. Since $\gap_n\approx C_w/n$, the total cost is $\lambda C_w/n+\kappa n$, minimized at $n^*=\sqrt{\lambda C_w/\kappa}$. The square root brings insensitivity: misestimating either price by $2\times$ moves $n^*$ by only $\sqrt2$.

\begin{table}[t]
\centering
\caption{The law's constants $C_w$ (SSIM) and the ladder sizes $n(\varepsilon)\approx C_w/\varepsilon$ needed to keep the quality gap within a tolerance $\varepsilon$.}
\label{tab:cw}
\small
\begin{tabular}{lccc}
\toprule
Configuration & $C_w$ & $n(0.01)$ & $n(0.005)$ \\
\midrule
Easy / Network 1    & 0.020 & 2 & 4 \\
Medium / Network 1  & 0.080 & 8 & 16 \\
Complex / Network 1 & 0.108 & 11 & 22 \\
Easy / Network 2    & 0.010 & 1 & 2 \\
Medium / Network 2  & 0.039 & 4 & 8 \\
Complex / Network 2 & 0.052 & 5 & 10 \\
\bottomrule
\end{tabular}
\end{table}

\subsubsection*{The law against ground truth}
Figure~\ref{fig:gapvsn} compares the prediction $C_w/n$ with exact optimal gaps for $n$ up to 64: the optimal-gap curves descend with slope $-1$ and converge onto their asymptotes from above. Quantitatively (Complex, Network 1), the ratio of true optimal gap to $C_w/n$ is $1.06$ at $n=8$, $1.03$ at $n=16$, and $1.008$ at $n=48$. By $n\approx8$ the law is an engineering-grade estimate, and the estimate is \emph{one integral}. The $1/n$ decay explains the slow returns practitioners observe: doubling the rung count halves the gap rather than quartering it. The Network-2 constants are \emph{smaller}---on a faster network more traffic sits where quality has saturated---yet the optimal rungs roughly double at every $n$, tracking the bandwidth quantiles. Table~\ref{tab:quantile} shows the exact optima for both networks alongside the law's prediction.

\begin{figure}[t]
\centering
\includegraphics[width=0.5\textwidth]{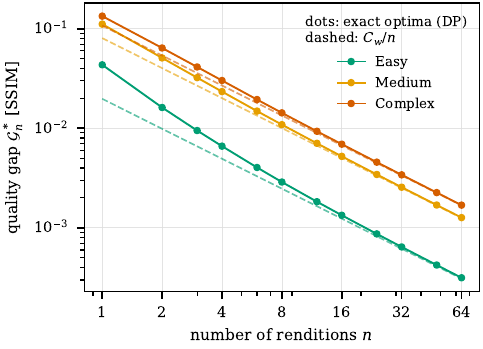}
\caption{The $\Theta(1/n)$ law against exact optima (Network 1, all contents). Dots: gaps of optimal ladders (dynamic programming); dashed: $C_w/n$ from Table~\ref{tab:cw}.}
\label{fig:gapvsn}
\end{figure}

\section{Extensions}
\label{sec:ext}

The analysis so far fixed one encoding resolution and took an objective metric at face value. We now remove both simplifications, and a third---that clients adapt to bandwidth alone---using the models of~\cite{Reznik2021DCC,Reznik2021PCS,Barman2022EUVIP}; Fig.~\ref{fig:ext} and Table~\ref{tab:models2} summarize their forms, parameters, and origin. How the solution survives these changes is itself a result: \emph{the mathematics is invariant; only the constants move.}

\begin{table}[!t]
\centering
\caption{The models used in Fig.~\ref{fig:ext} and the numerical examples of the extensions.}
\label{tab:models2}
\footnotesize
\begin{tabular}{@{}p{0.155\textwidth}p{0.45\textwidth}p{0.33\textwidth}@{}}
\toprule
Model & Form used in the examples & Origin and fitting \\
\midrule
Resolution-dependent quality $Q(H,R)$ & The quality--rate functions of Table~\ref{tab:models} with resolution-dependent parameter $\rho_c(H)=\alpha_c H^{\beta_c}$. ``Easy'': $\alpha_c=0.7844\cdot10^{-3}$, $\beta_c=1.2281$; ``Medium'': $\alpha_c=0.8278\cdot10^{-2}$, $\beta_c=1.3217$; ``Complex'': $\alpha_c=0.07316$, $\beta_c=1.0957$. & Fitted through a few dozen probe encodes per sequence; this model was shown to predict SSIM of H.264-encoded video reasonably well~\cite{Reznik2021DCC}. \\[2pt]
Separable perceptual quality model $S(H,H_p)\,Q(H,R)$ & Model~\eqref{eq:sep}: the factor $S(H,H_p)$ is derived from the Westerink--Roufs model~\cite{WesterinkRoufs1989}, evaluated at the window's viewing angle and the video's effective spatial frequency; the codec factor is a monotone remapping of the quality--rate model above onto the 1--5 MOS scale. & Proposed in~\cite{Reznik2021DCC}; fitted to subjective scores and shown to perform well across several datasets~\cite{Reznik2021DCC,Barman2022EUVIP}. \\[2pt]
Client resolution thresholds & Rule~\eqref{eq:clientH} with fitted $\alpha^*=0.723$. & Fitted to observed web-player behavior~\cite{Reznik2021PCS}. \\[2pt]
Web-audience window heights & Eleven observed heights $H_p$ with probabilities $q_{H_p}$: (228, .103), (240, .018), (380, .063), (430, .027), (480, .481), (630, .038), (678, .084), (710, .018), (774, .033), (810, .052), (990, .084); mean 538 lines. & Measured session statistics~\cite{Reznik2021PCS}. \\
\bottomrule
\end{tabular}
\end{table}

\begin{figure}[t]
\centering
\includegraphics[width=0.97\textwidth]{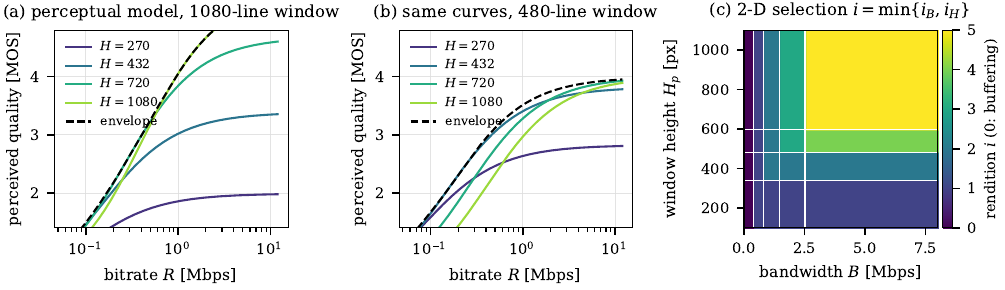}
\caption{The extensions in one view (Complex content). (a) Perceptual quality at four encoding resolutions in a 1080-line window: the curves cross, and the achievable frontier is their upper envelope (dashed). (b) The same encodings in a 480-line window: the window caps attainable quality---saturation. (c) The two-dimensional client rule $i=\min\{i_B,i_H\}$ for the five-rendition web-audience ladder of Table~\ref{tab:designs2d}: L-shaped cells bounded by two independent threshold chains (white lines); the dark band below $R_1$ is buffering.}
\label{fig:ext}
\end{figure}

\subsection{Resolutions and perceptual quality}
\label{sec:env}

Real ladders vary the encoding resolution $H$ because at any bitrate there is a best resolution to spend it on: at low rates a smaller, softer picture wins; at high rates it hits its detail ceiling. The quality--rate model becomes a family $Q(H,R)$---one curve per encoding resolution $H$, each satisfying the assumptions~\eqref{eq:codec}.

More broadly, we can also introduce a \emph{perceptual quality model} $Q(H,H_p,R)$: the quality of a video encoded at resolution $H$ and rate $R$, as seen in a window of height $H_p$ and with some other known parameters (display density, display brightness, distance between the viewer and the screen, etc.). Two properties of such a model carry all subsequent mathematics. The first is \emph{separability}:
\begin{equation}
Q(H,H_p,R) \;=\; S(H,H_p)\,Q(H,R),
\label{eq:sep}
\end{equation}
a product of a scaling and reproduction quality factor $S(H,H_p)$ and the codec-level model $Q(H,R)$. The second is \emph{saturation} due to the window limit: $S(H,H_p)=S(H_p,H_p)$ for $H\ge H_p$ (Fig.~\ref{fig:ext}(b)). The achievable quality--rate frontier for a window is the upper \emph{envelope} of the family, $Q^*(R)=\max_H Q(H,H_p,R)$, dashed in Fig.~\ref{fig:ext}(a,b). In all our computations we use the model proposed in~\cite{Reznik2021DCC}, with characteristics summarized in Table~\ref{tab:models2}. This model reports results using the 1--5 mean-opinion-score (MOS) scale, which we adopt in the presentation of all perceptual quality results.

\subsection{The envelope reduction}
\label{sec:envred}

Under a bandwidth-only client and a single viewing condition---a fixed window height $H_p$---the enlarged design problem collapses back onto the one just solved. The loads $p_i=F(R_{i+1})-F(R_i)$ depend on the rates alone, so the average quality $\Qbar=\sum_{i=1}^n p_i\,Q(H_i,H_p,R_i)$ is maximized over resolutions term by term, by $H_i\in\argmax_H Q(H,H_p,R_i)$, and the design problem reduces to~\eqref{eq:design} with the envelope $Q^*(R)$ in place of $Q(R)$. (Resolutions must be nondecreasing along the ladder; the maximizer is nondecreasing in $R$ for all fitted models used here, so this holds automatically.) Everything of the preceding sections then applies to $Q^*$: the identification, the dynamic program---run once on the envelope, with resolutions read off the argmax---and Proposition~\ref{thm:law}. The envelope is nondecreasing but only piecewise $C^1$---it has kinks where the curves cross---so the proposition applies after splitting $[0,b]$ at the finitely many kinks.

\subsection{Two-dimensional adaptation: player size}
\label{sec:2d}

On the web, video plays in a window of layout-determined height $H_p$, distributed broadly across sessions (Table~\ref{tab:models2}). Fetching more lines than the window can display wastes bits, so web players \emph{cap} their selection by the window. The fitted client model~\cite{Reznik2021PCS} adds a second staircase, mirroring the bandwidth rule~\eqref{eq:client} over the ladder's resolutions $H_1\le\dots\le H_n$:
\begin{equation}
i_H(H_p) = j \;\Longleftrightarrow\; T^H_{j-1} \le H_p < T^H_j,
\qquad
T^H_j=\alpha H_j+(1{-}\alpha)H_{j+1},
\label{eq:clientH}
\end{equation}
for $j=1,\dots,n$, with the conventions $T^H_0=0$ and $T^H_n=\infty$: each switching threshold is a weighted average of two consecutive ladder resolutions, with a fitted weight $\alpha\in[0,1]$. The two staircases compose by taking the more restrictive one,
\begin{equation}
i(B,H_p) = \min\{\, i_B(B),\; i_H(H_p) \,\},
\label{eq:client2d}
\end{equation}
so bandwidth proposes and the window caps. The audience splits by window height into \emph{audience segments}, each seeing the ladder truncated at its own cap---and the segments disagree about where rungs should go. Under the min-composition, rendition $i$ is delivered on an \emph{L-shaped} region bounded by two \emph{independent} families of lines: vertical at the rates $R_j$, horizontal at the thresholds $T^H_j$ (Fig.~\ref{fig:ext}(c)).

The average delivered quality now runs over both variables. With the window height $H_p$ taking its observed values with probabilities $q_{H_p}$, and with the selection rule~\eqref{eq:client2d},
\begin{equation}
\Qbar \;=\; \sum_{H_p} q_{H_p} \int_0^\infty p(B)\; Q\bigl(H_{i(B,H_p)},\,H_p,\,R_{i(B,H_p)}\bigr)\,\dd B,
\label{eq:Q2d}
\end{equation}
where buffering ($i=0$) contributes zero quality, as before. This looks like a genuinely two-dimensional design problem. It is not: the L-shaped cell geometry and the separability of the quality model collapse it.

\begin{theorem}[Exact reduction;~\cite{Reznik2021DCC}]
\label{thm:red}
For any separable quality model~\eqref{eq:sep}, the two-dimensional average quality~\eqref{eq:Q2d} equals
\begin{equation}
\Qbar = \sum_{i=1}^n Q(H_i,R_i)\Bigl[ A_i\bigl(F(t_i)-F(t_{i-1})\bigr) + C_i\bigl(1-F(t_{i-1})\bigr)\Bigr],
\label{eq:reduced}
\end{equation}
where $t_{i-1}=R_i$ for $i=1,\dots,n$, $t_n=\infty$, and the audience weights
\begin{equation}
A_i=\sum_{H_p:\, i_H(H_p)>i} q_{H_p}\,S(H_i,H_p),
\qquad
C_i=\sum_{H_p:\, i_H(H_p)=i} q_{H_p}\,S(H_i,H_p),
\label{eq:weights}
\end{equation}
collect the segments whose cap exceeds $i$ and the segments capped at $i$; both depend on the resolutions alone.
\end{theorem}

\begin{proof}
Write $k=i_H(H_p)$ for the cap of the segment with window height $H_p$---by~\eqref{eq:clientH} a function of the resolutions alone. Under the rule~\eqref{eq:client2d}, a client of this segment receives rendition $i<k$ exactly when $R_i\le B<R_{i+1}$, and rendition $k$ whenever $B\ge R_k$: below the cap the cells are the ordinary bandwidth cells, and everything above $R_k$ collapses into one capped tail. By separability~\eqref{eq:sep}, the quality delivered to this segment by rendition $i$ is $S(H_i,H_p)\,Q(H_i,R_i)$. Substituting into~\eqref{eq:Q2d} and splitting the inner integral at the rates,
\[
\Qbar=\sum_{H_p} q_{H_p}\Bigl[\,\sum_{i<k} S(H_i,H_p)\,Q(H_i,R_i)\,p_i \;+\; S(H_k,H_p)\,Q(H_k,R_k)\bigl(1-F(R_k)\bigr)\Bigr].
\]
Exchanging the sums over $H_p$ and $i$ and collecting the terms at each rendition gives~\eqref{eq:reduced} with the weights~\eqref{eq:weights}. Separability is what closes the audience sum into these finite weights: without the factorization~\eqref{eq:sep}, the sum over $H_p$ would remain inside the bandwidth integral. Both weights depend on the resolutions alone---through $H_i$ in $S$, and through the thresholds~\eqref{eq:clientH} in $i_H$.
\end{proof}

\emph{No two-dimensional integral survives}: bandwidth enters only through $F$ at the rates, player size only through the finite sums $A_i,C_i$, and the objective is chain-structured in the rates. Design becomes an outer enumeration over monotone resolution chains around the inner dynamic program of Section~\ref{sec:dp}. \begin{table}[t]
\centering
\caption{Optimal five-rendition two-dimensional ladders for two audiences (Complex content, Network 1, perceptual quality).}
\label{tab:designs2d}
\small
\begin{tabular}{ll}
\toprule
Audience & Optimal renditions $(H_i,\,R_i)$ [lines, kbps] \\
\midrule
Web audience (heights of Table~\ref{tab:models2}) & (288, 398), (480, 805), (480, 1463), (480, 2496), (900, 2550) \\
Full screen, 1080-line windows & (480, 283), (900, 739), (1080, 1402), (1080, 2391), (1080, 4443) \\
\bottomrule
\end{tabular}
\end{table}

The designs are visibly audience-specific (Table~\ref{tab:designs2d}): the web audience's optimum uses resolutions 288--900 only, with no 1080-line rung, and tops out at $2550$ kbps where the full-screen optimum reaches $4443$. Rungs follow the audience's joint mass in (window, bandwidth), not the display technology's maximum.

The asymptotics transfer conditionally. A single-segment audience sees the one-dimensional problem of the envelope reduction, so Proposition~\ref{thm:law} applies with the envelope constant ($C_w=1.02$ MOS for Complex content on Network 1), and $n\,\gap_n^*$ descends toward it from above, as in Fig.~\ref{fig:gapvsn}. For mixed audiences no single constant applies, but the range can be bracketed: if all rungs were shared by all segments, the audience-averaged loss density would give $n\gap_n\to0.49$ MOS; if every segment required private rungs, a Cauchy--Schwarz allocation over the segments' own law constants $C_w(H_p)$ gives $\bigl(\sum_{H_p}\sqrt{C_w(H_p)}\bigr)^2=4.1$ MOS. Measured exact optima run between these brackets, with $n\,\gap_n^*$ still growing through $n=10$ and an effective log--log slope of the gap around $-0.4$: for realistic web audiences, adding rungs pays even less than the $1/n$ law suggests, and \emph{which segment} the next rung serves matters more than adding it. The joint limit---segments, caps, and $n$ interacting---is the principal open problem.

\subsection{What else transfers}
The same argument covers more: richer client policies replace hard cells by kernel-weighted soft ones; mixtures of audiences or networks enter linearly through $p$; multi-codec ladders (H.264, HEVC, AV1)~\cite{Reznik2019ICME} segment the audience by decoder capability exactly as window size does, and Theorem~\ref{thm:red} extends term by term. In each case the identification, the dynamic program, and the $\Theta(1/n)$ law with recomputed constants are expected to carry over.

\section{Conclusion}

Ladder design turned out to be scalar quantization of the bandwidth distribution, of an edge-constrained, one-sided kind, and the classical high-resolution answers reshape accordingly: the optimal gap decays as $\Theta(1/n)$ with explicit constant $C_w$, rates should follow $\sqrt{p\,Q'}$, quantile designs land a few percent from optimal, and ``how many rungs?'' has a closed-form answer. The structure is not streaming's alone---it recurs in the assortment, approximation, and link-adaptation problems of the introduction, where its core asymptotics were, in deterministic form, first discovered. What this paper contributes is the connection: the identification, the probabilistic Zador form with its explicit constant, and the design calculus. The open problems we find most inviting are the Zador-general version of the law, including the atoms that rate-capped service tiers produce; converses under richer client models; and the joint multi-segment limit.

\end{document}